\documentclass[journal,twoside]{IEEEtran}

\usepackage{stfloats}

\usepackage{url}
\usepackage{ifthen}
\usepackage[cmex10]{amsmath} 
\usepackage{revsymb}
                 
\usepackage{xcolor}

\usepackage{csquotes}
\usepackage{bbm}
\usepackage{comment}
\usepackage{physics}
\usepackage{amssymb}
\renewcommand{\Tr}{{\operatorname{Tr}}}
\newcommand{\id}{\operatorname{id}}
\newcommand{\proj}[1]{\left|#1\right\rangle\!\left\langle#1\right|}
\newcommand{\1}{\openone}
\newcommand{\ox}{\otimes}

\usepackage{lipsum}
\usepackage{setspace}

\usepackage{graphicx}

\usepackage{amssymb}
\usepackage{amsthm}
\usepackage{comment}

\usepackage{color}

\usepackage[english]{babel}
\usepackage[T1]{fontenc}
\usepackage{mathrsfs}
\usepackage{braket}

\usepackage{enumerate}

\usepackage{mathtools}
\usepackage[export]{adjustbox}
\usepackage{microtype}

\theoremstyle{plain}

\usepackage[hidelinks]{hyperref}

\usepackage{orcidlink}

\newtheorem{theorem}{Theorem}

\newtheorem{lemma}[theorem]{Lemma}

\newtheorem{definition}[theorem]{Definition}  
  
\newtheorem*{remark*}{Remark}   
\renewcommand\qedsymbol{$\blacksquare$}
\newenvironment{proof-of}[1][{\hspace{-\blank}}]{{\medskip\noindent\textit{Proof~{#1}.\ }}}{\hfill\qedsymbol}

\begin{document}


\title{New Protocols for Conference Key\protect\\ and Multipartite Entanglement Distillation%
\thanks{Date: 25 February 2025. An earlier version of this work was presented at ISIT 2022 \cite{FSAW:GHZ-mixed-ISIT}. The present paper has full proof details, additional protocols and added discussion.}%
\thanks{FS is supported by 
a Walter Benjamin Fellowship, DFG project no.~524058134. FS also acknowledges support by the DFG Cluster of Excellence ``Munich Center for Quantum Science and Technology'' (MCQST). AW is supported by the European Commission QuantERA grant ExTRaQT (Spanish MICIN project PCI2022-132965), by the Spanish MICIN (projects PID2019-107609GB-I00 and PID2022-141283NB-I00)) with the support of FEDER funds, by the Spanish MICIN with funding from European Union NextGenerationEU (PRTR-C17.I1) and the Generalitat de Catalunya, by the Spanish MTDFP through the QUANTUM ENIA project: Quantum Spain, funded by the European Union NextGenerationEU within the framework of the ``Digital Spain 2026 Agenda,'' by the Alexander von Humboldt Foundation, and the Institute for Advanced Study of the Technical University Munich.}}

\author{Farzin Salek\orcidlink{0000-0002-4222-3654} and Andreas Winter\orcidlink{0000-0001-6344-4870}%
\thanks{Farzin Salek is with the Department of Mathematics, Technical University of Munich, and the Munich Center for Quantum Science and Technology (MCQST), Germany (email: farzin.salek@gmail.com).}%
\thanks{Andreas Winter is with ICREA, with the Grup d'Informació Quàntica, Departament de Física, Universitat Autònoma de Barcelona, Spain, with the Institute for Advanced Study, Technical University of Munich, Germany, and with the Quantum Information Independent Research Centre Kessenich (QUIRCK), Germany (email: andreas.winter@uab.cat).}}

\maketitle

\begin{abstract}
We approach two interconnected problems of quantum information processing in networks: Conference key agreement and entanglement distillation, both in the so-called source model where the given resource is a multipartite quantum state and the players interact over public classical channels to generate the desired correlation.  
The first problem is the distillation of a conference key when the source state is shared between a number of legal players and an eavesdropper; the eavesdropper, apart from starting off with this quantum side information, also observes the public communication between the players. 
The second is the distillation of Greenberger-Horne-Zeilinger (GHZ) states by means of local operations and classical communication (LOCC) from the given mixed state. 
These problem settings extend our previous paper [IEEE Trans. Inf. Theory 68(2):976-988, 2022], and we generalise its results: using a quantum version of the task of communication for omniscience, we derive novel lower bounds on the distillable conference key from any multipartite quantum state by means of non-interacting communication protocols. Secondly, we establish novel lower bounds on the yield of GHZ states from multipartite mixed states. Namely, we present two methods to produce bipartite entanglement between sufficiently many nodes so as to produce GHZ states.
Next, we show that the conference key agreement protocol can be made coherent under certain conditions, enabling the direct generation of multipartite GHZ states. 
\end{abstract}


\section{Introduction and preliminaries}
\IEEEPARstart{A}{pplications} of quantum networks to produce correlations among designated parties are  among the most mature in quantum information.
Conference key agreement in particular is a fundamental task whose objective is to allow multiple parties within a network to leverage quantum properties such as entanglement and superposition to establish a common secret key.
Considerable research effort has been devoted to the study of bipartite secret key and entanglement in quantum networks \cite{1362905,PhysRevA.72.052317,PhysRevLett.93.230504}.
Before delving further into the topic, let us first establish some necessary notation. Any additional conventions required will be introduced as we come across them in our discussion.
The capital letters $X$, $T$, etc., denote random variables, whose realizations are shown by the corresponding lowercase ($x$, $t$, etc.) and whose alphabets (ranges) are shown by calligraphic letters ($\mathcal{X}$, $\mathcal{T}$, etc.), respectively. Quantum systems $A$, $B$, etc., are associated with (finite-dimensional) Hilbert spaces denoted with the same letter, whose dimensions are denoted by $|A|$, $|B|$, etc. Multipartite systems $AB\ldots Z$ are described by tensor product Hilbert spaces $A\otimes B\otimes \cdots \otimes Z$.
For any positive integer $m$, we use the notation $[m]=\{1,\ldots,m\}$. For conciseness, we denote the tuple $(X_{1},\ldots,X_{m})$ by $X_{[m]}$, and similarly for (block) indices in superscript. Moreover, for a subset $J\subset[m]$ of indices, we write $X_{J}=(X_{j}: j\in J)$. 
Throughout the paper, $\log$ denotes by default the binary logarithm. The trace norm (aka Schatten or non-commutative $1$-norm) is $\|\omega\|_1 = \Tr\sqrt{\omega^\dagger\omega} = \max_{\|\Lambda\|_\infty\leq 1} \Tr\omega\Lambda$.
The purified distance between possibly sub-normalised quantum state $\rho$ and $\sigma$ is defined as $P(\rho,\sigma)=\sqrt{1-F(\rho,\sigma)^2}$, with the generalized fidelity $F(\rho,\sigma)=\left\|\sqrt{\rho}\sqrt{\sigma}\right\|_1 + \sqrt{(1-\Tr\rho)(1-\Tr\sigma)}$. Note that if at least one of the states is normalised, the fidelity reduces to its familiar form $F(\rho,\sigma)=\left\|\sqrt{\rho}\sqrt{\sigma}\right\|_1$. 

Extraction of keys from a pair of random variables $X_{1}$ and $X_{2}$ secret from another random variable $Z$ was studied by Maurer \cite{Maurer:key}. Ahlswede and Csisz\'{a}r \cite{243431} in particular introduced and solved the so-called one-way communication protocols, where only either the party holding $X_1$ or the one holding $X_2$ can broadcast a message over the public noiseless channel. The latter paper presented the optimal rate of this task, which is given by a single-letter expression involving the difference between certain conditional mutual information of the random variables and auxiliary random variables. 

The extensive development of applications of quantum networks involves using genuine multipartite quantum protocols, whose aim it is to share multipartite secret key and entanglement among many players 
\cite{Das:network,AH:mE,MGKB,Streltsov-et-al}. 
Secret key agreement in a classical network with $m$ players 
was studied by Csisz\'ar and Narayan \cite{Csiszar-Narayan.2004} using 
an approach called communication for omniscience (CO): 
$m$ players observe a correlated discrete memoryless multiple source 
$X_{[m]} = (X_{1},\ldots,X_{m})$, the $j$-th node obtaining
$X_{j}$. The nodes are allowed to communicate interactively over a public noiseless 
broadcast channel so that at the end they attain omniscience: each node reconstructs 
the whole vector of observations $X_{[m]}$. 
The objective is to minimise the overall communication to achieve this goal. The key observation was that players can achieve omniscience through non-interactive communications, wherein each player only needs to transmit a single message to others based on its local information.

Quantum systems can exhibit intricate correlations that cannot be fully understood using classical intuition alone.
Measuring the amount and type of non-classical correlations present in a quantum system, provides insights into the degree of quantum entanglement within the state. However, quantifying entanglement and characterizing its properties pose considerable challenges and require sophisticated mathematical tools and techniques. 
There are only a few instances where the entanglement content of a state is fully understood. One such example involves the asymptotic limit of many copies of a bipartite pure state $\ket{\psi}^{AB}$: not only can it be transformed into EPR states $\ket{\Phi} = \frac{1}{\sqrt{2}}(\ket{0}\ket{0}+\ket{1}\ket{1})$ at a rate of $E(\psi)=S(A)_\psi$ using local operations and classical communication (LOCC), but remarkably, the same rate governs the reverse transformation from $\phi$ back to $\psi$ \cite{PhysRevA.53.2046}. Here, $S(A)_{\rho}=-\Tr\rho^{A}\log\rho^{A}$ denotes the von Neumann entropy of the reduced state of the quantum state $\rho^{AB}$. Another illustrative example involves a tripartite stabilizer state $\ket{\psi}^{ABC}$ distributed among three remote parties, each holding multiple qubits. In a notable study \cite{BFG}, it was demonstrated that this state can be transformed by local unitaries into a combination of EPR pairs, GHZ states $\ket{\Gamma_3} = \frac{1}{\sqrt{2}}(\ket{0}\ket{0}\ket{0}+\ket{1}\ket{1}\ket{1})$, and local one-qubit states. The quantities of EPR and GHZ states depend on the dimensions of specific subgroups within the stabilizer group. The authors further provide a formula for determining the maximum number of tripartite GHZ states that can be extracted from $\ket{\psi}^{ABC}$ through the use of local unitaries.

The picture becomes significantly less clear when dealing with mixed states \cite{RevModPhys.81.865}, as even a complete understanding of classical correlations in $\rho^{AB}$ is lacking.
A notable exception arises in the extraction of secret keys from $\rho^{AB}$ under the so-called one-way communication protocols, for which a formula, though multi-letter, is known: by generalizing the one-way communication protocol of Ahlswede and Csisz\'{a}r \cite{243431} to bipartite quantum states in \cite{2005RSPSA.461..207D},
the authors formulated the optimal rate of keys distillable from $\rho^{ABE}$ by one-way communication between users holding systems $A$ and $B$ secret from the eavesdropper who holds system $E$, and the result is given as the regularization $K_\rightarrow(\rho) =\lim_{n\rightarrow\infty} \frac1n K^{(1)}_\rightarrow\left(\rho^{\otimes n}\right)$ of the following single-letter formula \cite[Thm.~8]{2005RSPSA.461..207D}:
\begin{align*}
  K^{(1)}_\rightarrow(\rho) = \max_{Q,T-X} \left\{I(X;B|T)-I(X;E|T)\right\}.
\end{align*}
Here, the maximisation is over all POVMs $Q=\{Q_{x}\}_{x\in\mathcal{X}}$ on $A$
and classical channels $r(t|x)$, and $I(X:B|T)$ and $I(X:E|T)$ are conditional quantum mutual informations of the state 
\[\begin{split}
\omega^{TXBE} 
 =\sum_{t,x}r(t|x) \ketbra{t}^{T} &\otimes\ketbra{x}^{X} \\ &\otimes\Tr_{A}\rho^{ABE}\left(Q_{x}\otimes\1_{BE}\right)
\end{split}\]
and its marginals:
\[\begin{split}
  I(X:B|T)_\omega 
  &= S(XT)_\omega + S(BT)_\omega - S(T)_\omega - S(XBT)_\omega.  
\end{split}\]
In the same paper, it was demonstrated that, under certain additional conditions, the key distillation protocol can be made coherent (this aspect will be discussed extensively in Sec.~\ref{make-coherent}). As a result, the paper establishes the following achievable rate of EPR pairs from any bipartite state $\rho^{AB}$, known as the coherent information:
\begin{align*}
    I(A\rangle B)_\rho = -S(A|B) = S(B)_\rho - S(AB)_\rho,
\end{align*}
which is the negative conditional quantum entropy.
This single-letter expression holds considerable importance in various fields of quantum information. Regrettably, akin to the secret key rate, a regularized expression represents the highest rate of distillable EPR pairs from any bipartite state. 
Specifically, consider an instrument $\mathcal{E}=\{\mathcal{E}_x\}_x$ on system $A$, where $\mathcal{E}_x$ are completely positive (cp) maps sending $A$ to the joint output of quantum system $A'$ and classical system $X$, i.e. $\mathcal{E}_x:A\mapsto A'\otimes X$, such that their sum is trace-preserving. For any bipartite state $\rho^{AB}$, the one-way distillable entanglement can be expressed as the regularization 
\(
    D_\rightarrow(\rho) = \lim_{n\rightarrow\infty}\frac{1}{n}D^{(1)}_\rightarrow(\rho^{\ox n}),
\)
with
\begin{align*}
    D^{(1)}_\rightarrow(\rho)\coloneqq \max_{\mathcal{E}}\sum_{x\in\mathcal{X}}p(x)
    I(A\rangle B)_{\rho_x^{A'B}},
\end{align*}
where the maximization is over quantum instruments 
$\mathcal{E}=\{\mathcal{E}_x\}_x$ and
\begin{align*}
    p(x) = \Tr\mathcal{E}_x(\rho) \quad \text{and} \quad \rho_x^{A'B}
    = \frac{1}{p(x)}(\mathcal{E}_x\ox \id_B)\rho^{AB}.
\end{align*}
While it is possible to bound the range of $X$ by $|A|^2$, and each cp map can be assumed to have only one Kraus operator, the expression is not generally computable due to the infinite copy limit involved.
In the same paper, the authors also proved another multi-letter formula for the two-way communications secret key rate. 

A very different setting for the extraction of EPR pairs is one where all players assist two distinguished players through LOCC to distill EPR pairs between themselves.
This problem was initially explored in \cite{1998quant.ph..3033D} under the name of \emph{entanglement of assistance} for a pure initial state. In the asymptotic regime, a formula was discovered for up to $4$ parties in \cite{PhysRevA.72.052317}, and subsequently, extended to an arbitrary number of parties in \cite{state-merging}, always for a pure initial state.
The optimal rate for this problem takes a particularly simple form given by the minimum-cut bipartite entanglement: given  the pure state $\ket{\psi}^{A_1\ldots A_m}$, the entanglement of assistance between two parties $A_i$ and $A_j$ equals 
\[
  E_A(A_i:A_j|\psi) = \min_{I \subseteq [m]\setminus\{i,j\}} S(A_{i}A_I)_\psi.
\]

The setting was later extended to mixed states in \cite{Dutil-Hayden}, where lower and upper bounds are derived. We will delve into this case in greater detail in Sec. \ref{merging}. Another problem concerns the distillation of EPR pairs between a distinguished party and the rest of the parties. This particular problem is referred to as \emph{entanglement combing} and has been fully solved for the case of pure states in \cite{PhysRevLett.103.220501}. In Sec. \ref{merging}, we explore a generalization of this problem to mixed states by employing the mixed state entanglement of assistance \cite{Dutil-Hayden,Dutil:PhD}.

The generalization of the CO problem to quantum networks, where the $j$-th
player instead of a random variable $X_j$ observes the subsystem $A_j$ of 
a multipartite quantum state $\rho^{A_1\ldots A_m}$, 
leads quite naturally to the problem of existence of a simultaneous decoder for the 
classical source coding with quantum side information at the decoder \cite{WINTER:PHD}.
By finding a simultaneous decoder, the present authors generalized the CO
problem to quantum networks and derived novel lower bounds on the distillation of common randomness 
(CR) and Greenberger-Horne-Zeilinger (GHZ) states from multipartite quantum states \cite{GHZ:ISIT2020,FSAW:GHZ}. 
More precisely, we studied distillation of CR from mixed states $\rho^{A_1\ldots A_m}$
and GHZ states from pure multipartite states $\ket{\psi}^{A_1\ldots A_m}$.
The present work generalizes the results of the former paper in two directions: 
first, we construct new protocols for distillation of CR secret against an eavesdropper 
with quantum side information,
i.e. a secret key shared by $m$ players when they have access to many copies of a 
mixed state $\rho^{A_1\ldots A_mE}$, where subsystem $E$ is the quantum information provided to an eavesdropper. Secondly,
we devise two distinct protocols employing state merging and entanglement of assistance to convert any quantum states into EPR pairs between different parties, subsequently enabling their transformation into GHZ states. This line of investigation yields two novel lower bounds. Additionally, we show, in a constructive manner, how the conference key agreement protocol can be made coherent under additional conditions, leading to the creation of GHZ states at novel rates.
Almost all of our protocols are from a subclass of protocols called \emph{non-interactive communication}, because each player broadcasts only one message that depends only on their local measurement data. 
 

\section{Conference key distillation}
%
%
%
%

Here we consider \emph{secret key distillation} in the source model. This means that we have $m+1$ separated players sharing $n \gg 1$ copies of an $(m+1)$-partite quantum state $\rho^{A_{1}\ldots A_{m}E}$, i.e.~legitimate player $j\in[m]$ holds the subsystem $A_j^n$ and the eavesdropper holds the subsystem $E^{n}$.
All legitimate players can communicate to each other through a public noiseless classical broadcast channel of unlimited capacity, but the eavesdropper gets a copy of all these communications. The most general definition of the secret key agreement protocols in the bipartite case was given in \cite{renner:phd,Wilmink2003} and in an important and different form in \cite{pbit-PRL,pbit-IEEE}.
In the multipartite case see \cite[Def.~5]{GHZ:ISIT2020} and \cite{FSAW:GHZ}, which are concerned with common randomness distillation rather than secret key agreement, but a secrecy condition can be added easily (see below).
We now define the most general non-interactive protocol for distilling an $(m+1)$-partite state $\rho^{A_1\ldots A_mE}$ into a secret key between $m$ players.

\begin{definition}
\label{cq-code}
A code (or protocol) for non-interactive secret key agreement consists of the following:
\begin{enumerate}
\item An instrument consisting of cp maps $\mathcal{E}^{(j)}_{\ell_j}:A_j^n\rightarrow A'_j$ for each player $j\in[m]$ {acting on $n$ blocks of $A_j$}  
(with quantum and classical registers $A'_j$ and $X_j$, respectively);
\item A POVM {acting on $A'_j$} $(D_{k_j}^{(j,\ell_{[m]})}:k_j\in\mathcal{K})$, for each player $j\in[m]$ and $\ell_{[m]}\in\mathcal{L}_1\times\cdots\times\mathcal{L}_m$.
\end{enumerate}
The idea is that each player $j$ applies their instrument to their share $A_j^n$ of the $n$ copies of initial multipartite mixed state, and broadcasts the outcome $\ell_j$. After receiving all messages from the other players, so that they all share knowledge of $\ell_{[m]}=(\ell_1,\ldots,\ell_m)$, each player $j$ will measure the POVM $D^{(j,\ell_{[m]})}$. 
The resulting state of the key for all players and the side information of the eavesdropper, the latter holding $L_{[m]}$ and $E^n$, is then
\begin{align*}
   &\Omega^{K_1\ldots K_mL_{[m]}E^n} \\
   &\phantom{.}
    =\!\!\! \sum_{k_1\ldots k_m \atop \ell_1\ldots\ell_m} \proj{k_1}^{K_1} \ox \cdots \ox \proj{k_m}^{K_m} \ox \proj{\ell_{[m]}}^{L_{[m]}} \\
   &\phantom{====}
    \ox \Tr_{A_{[m]}'} \left( (\mathcal{E}^{(1)}_{\ell_1}\ox\cdots\ox\mathcal{E}^{(m)}_{\ell_m}\ox\id_{E^n})\rho^{A_1^n\ldots A_m^nE^n}\right)\!\! \\
    &\phantom{==========} \times \left(D^{(1,\ell_{[m]})}_{k_1}\ox\cdots\ox  D^{(m,\ell_{[m]})}_{k_m}\ox\1_{E^n}\right)\!.
\end{align*}

For technical reasons we assume that communication of the $j$-th player has a rate $R_j$, i.e. $|\mathcal{L}_j|\leq 2^{nR_j}$, for some constants $R_j$.
We call this an $(n,\varepsilon)$-protocol if 
\begin{align}
  \label{reliability}
  \Pr\{K_1=\ldots =K_m\} &\geq 1-\varepsilon,\\
\label{secrecy}
 \frac{1}{2}\left\|\Omega^{K_1L_{[m]}E}-u_{K_1}\otimes\tau_{0}^{L_{[m]}E^n}\right\|_{1} &\leq \varepsilon,
\end{align}
where $u_{K_1}=\frac{1}{|\mathcal{K}|}\sum_{k_1}\ketbra{k_1}$ and $\tau_{0}$ is some constant state.

We call $R$ an achievable rate if for all $n$ there exist $(n,\varepsilon)$-protocols 
with $\varepsilon\rightarrow 0$ and $\frac{1}{n}\log|\mathcal{K}|\rightarrow R$. Finally we define the non-interactive secret-key capacity of $\rho$ as
\begin{align*}
K_{\text{n.i.}}(\rho) \coloneqq \sup\{R: R \text{ achievable}\}.
\end{align*}
\end{definition}

The restriction to non-interactive communication is supposed to simplify the problem, while providing some added generality with respect to one-way distillation protocols, but Definition \ref{cq-code} is still too general for us to handle. To state the following results, we consider a subclass of protocols where first each player $j$ applies the same instrument $(\mathcal{E}^{(j)}_{x_j}:A_j\rightarrow A_j')$ to each copy of their $n$ systems, with outputs $x_j^n$, and then broadcast a message $\ell_j$ that is a function only of $x_j^n$, by way of classical channels (stochastic maps) $T_j:\mathcal{X}_j^n\rightarrow\mathcal{L}_j$. The first step gives rise to a cq-state 
\begin{align}
  \label{eq:omega}
  &\omega^{X_1A_1'\ldots X_mA_m'E} \nonumber \\ 
  &\phantom{==}
   = \sum_{x_1\ldots x_m} \proj{x_1}^{X_1}\ox\cdots\ox\proj{x_m}^{X_m} \nonumber \\
   &\phantom{========}\ox (\mathcal{E}^{(1)}_{x_1}\ox\cdots\ox\mathcal{E}^{(m)}_{x_m}\ox\id_{E})\rho^{A_1\ldots A_mE},
\end{align}
where the registers $X_jA_j'$ are held by player $j$.

\subsection{Non-interactive conference key distillation protocol}
In the following, we prove a new achievability result for the distillable secret key from $\rho^{A_{1}\ldots A_{m}E}$.

\begin{theorem}
\label{secret-key cq}
With the notation above, for every $(m+1)$-partite state $\rho^{A_1\ldots A_mE}$, and the ensuing cq-state $\omega$ in Eq.~\eqref{eq:omega}, 
\begin{align*}
  K_{\text{n.i.}}(\rho) \geq S(X_{[m]}|E)_{\omega} - R_{\text{CO}}^{\text{cq}},
\end{align*}
where $R_{\text{CO}}^{\text{cq}}=\min_{R_{[m]}\in\mathcal{R}_{\text{cq}}} \sum_{j=1}^{m}R_{j}$ and $\mathcal{R}_{\text{cq}}$ 
is the set of the rate tuples $R_{[m]}=(R_{1},\ldots,R_m)$ satisfying 
\begin{align*}
\forall j\in[m]\ \forall J\subseteq[m]\setminus j\quad \sum_{i\in J}R_i\geq S(X_J|X_{[m]\setminus J},A'_j)_{\omega}.
\end{align*}
\end{theorem}

A special case of this theorem \cite[Thm.~4]{FSAW:GHZ} corresponds to a scenario where each party performs a full measurement. In this case, the corresponding rate region is similar, with the $A'_j$ systems removed (see Theorem \ref{secret-key c}).

The proof of the theorem rests on two main pillars, one is multiple source coding with quantum side information, and the other privacy amplification \cite{privacy-amplification:Bennett}. In the following we will briefly review the necessary definitions and properties and then prove Theorem \ref{secret-key cq}.

Let $m$ players share $n$ copies of $\rho^{A_1\ldots A_m}$.
Each player applies its instrument $\mathcal{E}^{(j)}:A_j\rightarrow A'_j\otimes X_j$ with classical $X_j$ and quantum $A'_j$ outputs, to its share of the initial multipartite 
mixed state turning it into a cq-state
$\omega^{X_{1}A'_1\ldots X_{m}A'_m}$. They players want to attain omniscience $X_{[m]}$ at all nodes. The following theorem supplies a lower bound for this task:
\begin{lemma}[{Salek~{\&}~Winter~\cite[Thm.~4]{FSAW:GHZ}}]
\label{cq-omniscience}
An inner bound on the optimal rate region for the CO problem is the set of rate tuples $R_{[m]}=(R_1,\ldots,R_m)$ satisfying
\begin{align*}
  \forall j\in[m]\ \forall J\subseteq[m]\setminus j\quad \sum_{i\in J}R_i\geq S(X_J|X_{[m]\setminus J},A'_j)_{\omega}.
\end{align*}
\end{lemma}

Hashing the omniscience rate down relies on privacy amplification \cite{privacy-amplification:Bennett} and its reigning entropy, the smooth min-entropy, 
of which we will briefly review the necessary definitions and properties; cf.~\cite{Tomamichel:book} for more details. 

\begin{definition}[{Cf.~\cite[Def.~6.2]{Tomamichel:book}}]
For a (possibly subnormalized) state $\rho^{AB}$, the min-entropy of $A$ conditioned on $B$ is defined as
\begin{align*}
H_{\text{min}}(A|B)_{\rho}=\max \lambda \text{ s.t. } \rho^{AB}\leq 2^{-\lambda}\1\otimes\sigma^{B},\
\end{align*}
where $\sigma^{B}$ is a (possibly subnormalized) density operator.
\end{definition}

\begin{definition}[{Cf.~\cite[Def.~6.9]{Tomamichel:book}}]
Let $\varepsilon\in[0,1)$ and $\rho^{AB}$ be a (possibly sub-normalized) state. The smooth min-entropy of $A$ conditioned on $B$ is defined as 
\begin{align*}
H_{\text{min}}^{\varepsilon}(A|B)_\rho 
=\max H_{\text{min}}(A|B)_{\rho'} \text{ s.t. } \rho' \stackrel{\varepsilon}{\approx} \rho,
\end{align*} 
where $\rho' \stackrel{\varepsilon}{\approx} \rho$ means $P(\rho,\rho')\leq\varepsilon$ for a (possibly subnormalized) state $\rho'$.
\end{definition}

The smooth min-entropy satisfies the asymptotic equipartition property (AEP) for $0<\varepsilon<1$,
\begin{align}
\label{eq:aep}
\lim_{n\rightarrow\infty}\frac{1}{n}H_{\text{min}}^{\varepsilon}(A^n|B^n)_{\rho^{\otimes n}} = S(A|B)_\rho,
\end{align}
as well as the following chain rule for $\rho^{AYB}$ with a classical register $Y$:
\begin{align}
\label{eq:chainrule}
H_{\text{min}}^{\varepsilon}(A|YB)_\rho\geq H_{\text{min}}^{\varepsilon}(A|B)_\rho-\log|Y|.
\end{align}

Consider a source that outputs a random variable $Z$ (to be identified as a classical system $Z$)
about which there exists quantum side information $E$ for an eavesdropper, jointly described by the cq-state $\rho^{ZE} = \sum_z p(z) \proj{z}^Z \ox \rho_z^E$. 
Privacy amplification concerns the question of how much uniform randomness [$K(\varepsilon)$ bits] can be extracted from $Z$ such that it is independent of the side information $E$ \cite{privacy-amplification:RK}, up to trace distance $\varepsilon$ from this ideal.
\begin{lemma}[{Cf.~\cite[Thm.~7.9]{Tomamichel:book}}]
\label{privacy-amplification}
Let $\varepsilon\in(0,1)$. The maximum number of bits of
uniform and independent randomness extractable from $\rho^{ZE}$ is lower bounded as
\begin{align*}
\log K(\varepsilon)\geq H_{\text{min}}^{\varepsilon'}(Z|E)_\rho-2\log\frac{1}{\delta},
\end{align*}
for any $\delta\in(0,\varepsilon)$ and $\varepsilon'=\frac{\varepsilon-\delta}{2}$.
\end{lemma}

\begin{proof}[Proof of Theorem \ref{secret-key cq}]
The idea is that each player applies its instrument independently to each copy of its share of the initial multipartite 
mixed state turning it into a cq-state
\begin{align*}
 \omega^{X_{1}A'_1\ldots X_{m}A'_mE} 
  &=(\mathcal{E}^{(1)}\otimes\cdots\otimes\mathcal{E}^{(m)})\rho\\
  &=\sum_{x_1,\ldots,x_m} \ketbra{x_{[m]}}^{X_{[m]}} \\
  &\phantom{====} \otimes (\mathcal{E}_{x_1}^{(1)}\otimes\cdots\otimes\mathcal{E}_{x_m}^{(m)})\rho^{A_1\ldots A_mE}.
\end{align*}
From $n$ copies of the initial mixed state, the protocol reduces to key extraction 
from $n$ copies of the cq-state $\omega$, where each player broadcasts
a {deterministic} function of its local data to other players over noiseless broadcast channel. 
After players broadcast $\ell_{[m]}$, the state transforms into
\begin{align*}
  \widehat{\omega}^{L_{[m]}^{m+1}A^{\prime n}_{[m]}E^n}
  &= \sum_{x_{[m]}^n,\ell_{[m]}}
\ketbra{\ell_{[m]}}^{\otimes (m+1)} \\
  &\phantom{=====}
   \otimes \ketbra{x_{[m]}}^{X_{[m]}} \otimes \widehat{\omega}^{A^{\prime n}_{[m]}E^n}_{(x_{[m]}^n,\ell_{[m]})},
\end{align*}
where the first $(m+1)$ registers, one belonging to each $m$ players and one to eavesdropper, indicate that everyone including eavesdropper know the 
broadcast information. Note that here, the
$\widehat{\omega}^{A^{\prime n}_{[m]}E^n}_{(x_{[m]}^n,\ell_{[m]})}$ are not normalized, rather the sum of their traces is $1$.
All players measure their version of the key $K_j$ ($j\in[m]$); the resulting state of the key for each player, say player $1$, and eavesdropper becomes
\begin{align*}
\Omega^{K_1L_{[m]}E}
 &= \sum_{x_{[m]}^n,\ell_{[m]},k_1}\ketbra{k_1}\otimes\ketbra{\ell_{[m]}}
\otimes\widetilde{\omega}^{E^n}_{(x_{[m]}^n,\ell_{[m]})},
\end{align*}
where 
$\widetilde{\omega}^{E^n}_{(x_{[m]}^n,\ell_{[m]})} = \Tr_{{A'{}}^n}D_{k_1}^{\ell_{[m]}}\widehat{\omega}^{A^{\prime n}_{[m]}E^n}_{(x_{[m]}^n,\ell_{[m]})}$ is the non-normalized post-measurement state.
We call this an $(n,\varepsilon)$-protocol if 
\begin{align}
  \label{reliability-1}
  \Pr\{K_1=\ldots =K_m\} \geq 1-\varepsilon,\\
\label{secrecy-2}
 \frac{1}{2}\left\|\Omega^{K_1L_{[m]}E}-u_{K_1}\otimes\tau_{0}^{L_{[m]}E^n}\right\|_{1}\leq \varepsilon,
\end{align}
where $u_{K_1}=\frac{1}{|\mathcal{K}|}\sum_{k_1}\ketbra{k_1}$ and $\tau_{0}$ is some constant state.

Following the protocol of Lemma \ref{cq-omniscience}, for block length $n$ the players broadcast a total of $nR_{\text{CO}}^{cq}+o(n)$ 
bits of information $\sum_{j=1}^{m}\ell_{j}$
to reach omniscience, i.e. share the state
\begin{align*}
\omega'
=\sum_{x_1^n,\ldots,x_m^n}
p(x_{[m]}^n)\ketbra{x_{[m]}^n}^{X_{[m]}^n}
\otimes \omega_{x_{[m]}^n}^{\prime L_{[m]}E^n},
\end{align*}
where all players traced out their residual quantum systems after reaching omniscience and systems 
on $L_{[m]}$ and $E^n$ denote the eavesdropper's classical and quantum side information, respectively.
In conformity with privacy amplification Lemma \ref{privacy-amplification}, the $m$ legal players can extract 
$H_{\text{min}}^{\varepsilon}(X_{[m]}^n|L_{[m]}E^n)_{\omega'}$ bits of uniform and independent
randomness. Applying the chain rule for the smooth min-entropy \eqref{eq:chainrule} and the asymptotic equipartition property \eqref{eq:aep}, the extracted key has length
\begin{align*}
 nR &\geq H_{\min}^\varepsilon(X_{[m]}^n|L_{[m]},E^n)_{\widehat\omega} \\
    &\geq H_{\min}^\varepsilon(X_{[m]}^n|E^n)_{\omega^{\ox n}} - \sum_i \log|\mathcal{L}_i| \\
    &\geq n S(X_{[m]}|E)_\omega - nR_{\text{CO}}^{cq} - o(n).
\end{align*}
This concludes the proof.
\end{proof}

A special case of this protocol is when each players applies a POVM (instead of an instrument) to its share of the initial multipartite 
mixed state turning it into a cq-state
\begin{align}
\label{cq-state}
\omega' 
 &=\!\!\! \sum_{x_1,\ldots,x_m} \!\!
\Tr\rho(M_{x_{1}}^{1}\otimes\cdots\otimes M_{x_{m}}^{m}) \ketbra{x_{[m]}}^{X_{[m]}}\otimes \rho_{x_{[m]}}^{E},
\end{align}
where $p(x_{[m]}) = \Tr\rho(M_{x_{1}}^{1}\otimes\cdots\otimes M_{x_{m}}^{m})$ is the joint distribution of $m$ random variables $\{X_i\}_{i=1}^m$ 
recording the measurement outcomes on $\rho^{A_1\ldots A_m E}$.
%
The following is a special case of Theorem \ref{secret-key cq}, in which the $m$ players apply full measurement instead of instruments.
\begin{theorem}\label{secret-key c}
With the notation above, for every $(m+1)$-partite state $\rho^{A_1\ldots A_mE}$,
\begin{align*}
K(\rho)_{\text{n.i.}}\geq S(X_{[m]}|E)_{\omega'} - R_{CO}^{c},
\end{align*}
where $R_{\text{CO}}^{\text{c}}=\min_{R_{[m]}\in\mathcal{R}_{\text{c}}} \sum_{i=1}^{m}R_{i}$ and $\mathcal{R}_{\text{c}}$ 
is the set of the rate tuples $R_{[m]}=(R_{1},\ldots,R_m)$ satisfying 
\begin{align*}
\label{omni:conditions}
\forall I\subsetneq[m]\quad \sum_{j\in I}R_j\geq H(X_I|X_{[m]\setminus I})_{\omega'}.
\end{align*}
\end{theorem}


\subsection{Non-optimality of omniscience protocols}
\label{subsec:n-i-limits}
One might wonder about the optimality of the key rate in Theorem~\ref{secret-key cq} (a question that presented itself already in our predecessor work \cite{GHZ:ISIT2020,FSAW:GHZ}). Evidently, one should optimise over local instruments $\mathcal{E}^{(i)}$, and presumably also allow regularisation (working directly with $n$ copies $\rho^{\ox n}$). 
Seeing that we consider only non-interactive LOCC protocols, it is natural to restrict the supposed converse to non-interactive protocols (that this is a serious restriction can be seen from specifically constructed examples, cf. \cite{Maurer:key} and \cite{Wilmink2003}). Looking at the classical case is then encouraging, as Csisz\'ar and Narayan have shown that the maximum common randomness rate distilled is indeed produced by communication for omniscience \cite{Csiszar-Narayan.2004}. One way to see this is to realize first that any protocol creating common randomness can be supplemented by additional communication to achieve omniscience of the original data vector $X_{[m]}$, while not decreasing the rate of common randomness created. 

However, in the quantum case we are going to argue that even among non-interactive protocols, the rate of Theorem~\ref{secret-key cq}, based on omniscience of the classical information generated before the first communication, is not optimal. Observe that in the classical setting, omniscience is uniquely defined because classical information, and only classical information, is there from the start. On the other hand, LOCC protocols generate correlated randomness as they go along. 

For this purpose, consider states of the form 
\begin{equation}\begin{split}
  \label{eq:against-CO}
  \rho^{ABCE} 
    &= \frac{1}{dk^3}\sum_{x=1}^d\sum_{\alpha,\beta,\gamma=1}^k 
\left(U_\alpha\proj{x}U_\alpha^\dagger\ox\proj{\beta}\right)^{A}\\
    &\phantom{===========}
    \ox \left(V_\beta\proj{x}V_\beta^\dagger\ox\proj{\gamma}\right)^{B}\\      
    &\phantom{==}
     \ox \left(W_\gamma\proj{x}W_\gamma^\dagger\ox\proj{\alpha}\right)^{C}\ox \proj{\alpha\beta\gamma}^E, 
\end{split}\end{equation}
where $U_\alpha$, $V_\beta$ and $W_\gamma$ are unitaries. 
Evidently, if each player measures their $\beta$,
$\gamma$ and $\alpha$, respectively, and broadcasts it, then each can
undo the local unitary $U_\alpha$, $V_\beta$ and $W_\gamma$, respectively,
and end up sharing the perfect secret key $x$, amounting to a rate 
$\log d$. This is also optimal, since conditional on Eve's knowledge 
of $\alpha\beta\gamma$, the local entropies are $\log d$, putting 
an upper bound on any distillable secret correlation. 
On the other hand, any protocol of omniscience as we consider would 
w.l.o.g. measure and broadcast $\beta$, $\gamma$ and $\alpha$, respectively,
as it is information Eve has anyway, and in this way all players share
it. However, at least one of the players would have to measure 
a significant part of the encrypted $x$-information to generate 
the local randomness that eventually goes into the omniscience 
information. For concreteness, let $k=2$ and $U_1=V_1=W_1=\1$, 
$U_2=V_2=W_2=\text{QFT}_d$, the quantum Fourier transform which maps the computational 
basis $\{\ket{x}\}$ to an unbiased basis. By the Maassen-Uffink entropic uncertainty relation \cite{MaassenUffink}, any measurement of any local system, producing a random variable $Y$, is constrained by $I(Y;X)\leq \frac12\log d$, 
and this is then an upper bound on the common randomness that player can distill with the other two taken together (because it is an upper bound on the Holevo information between $Y$ and the other two players). More generally, $I(Y:X^n) \leq \frac12 n\log d$ for $n$ independent copies of $\rho$ and an arbitrary measurement on any one party with outcome $Y$. 
Overall, our measure-and-communicate-for-omniscience (MCO) protocols cannot get above rate $\frac12\log d$ -- which incidentally is also 
achievable.

\section{GHZ distillation from mixed states}
We now move on to the distillation of entanglement in the form of GHZ states
\(
  \ket{\Gamma_m} = \frac{1}{\sqrt{2}}\left(\ket{0}^{\ox m} + \ket{1}^{\ox m}\right)
\)
from an $m$-partite mixed state $\rho^{A_1\ldots A_m}$. We present two distinct approaches for tackling this problem. 
The first is to use the multipartite resource state to produce bipartite entanglement in the form of EPR pairs between designated pairs of players, assisted by the others using a general LOCC procedure, and then to use the network of EPR pairs to generate GHZ states. We show two methods of doing this, one using quantum state merging in a generalisation of state combing \cite{PhysRevLett.103.220501}, the second using assisted entanglement distillation. 
The second approach rests on making coherent approach the secret key agreement protocol of Theorem \ref{secret-key cq}, applied to a purification of the mixed state $\rho^{A_1\ldots A_m}$.

We start by recalling the most general LOCC protocol for GHZ distillation, which is any $m$-partite channel $\Lambda$ acting on $A_1^n\ldots A_m^n$ that can be implemented by local operations and classical communications. It is called an $(n,\varepsilon)$-protocol for GHZ distillation
with rate $k/n$ if it acts on $n$ copies of the state $\rho^{A_1\ldots A_m}$ and produces $k$ copies of GHZ state $\ket{\Gamma_m}$ up to fidelity $1-\varepsilon$:
\begin{align*}
F\left( \ketbra{\Gamma_m}^{\otimes k}, \Lambda(\rho^{\ox n})\right) \geq 1-\varepsilon.
\end{align*}
A number $R$ is an achievable rate if for every $n$ there exist $(n,\varepsilon)$-protocols, with $\varepsilon\rightarrow0$ and $k/n\rightarrow R$ 
as $n\rightarrow\infty$. 
Then the GHZ distillation capacity of $\rho$ is defined as
\begin{align}
D(\rho)\coloneqq\sup\{R:R\ \text{achievable}\}.
\end{align}

The most general non-interactive protocol for GHZ distillation instead looks like this:

\begin{definition}
A non-interactive LOCC protocol consists of the following:
\begin{enumerate}
\item An instrument $\mathcal{E}^{j} = (\mathcal{E}^{j}_{\ell_j})_{\ell_j\in\mathcal{L}_j}$, for each player $j\in [m]$;
\item A quantum operation $\mathcal{G}^{(j)}_{\ell_{[m]}}$, 
for each player $j\in [m]$ and every public message tuple $\ell_{[m]} \in \mathcal{L}_{[m]}$.
\end{enumerate}
It is called an $(n,\varepsilon)$-protocol for GHZ distillation
with rate $k/n$ if it acts on $n$ copies of the state $\rho^{A_1\ldots A_m}$
and produces $k$ copies of the GHZ state $\ket{\Gamma_m}$ up to fidelity $1-\varepsilon$:
\begin{align}
F\left(
\ketbra{\Gamma_m}^{\otimes k},
\sigma^{B_1^k\ldots B_m^k}\right)
\geq 1-\varepsilon,
\end{align}
where 
\begin{align*}
\sigma^{B_1^k\ldots B_m^k} 
&= \sum_{\ell_{[m]}}
\left(\mathcal{G}^{1}_{\ell_{[m]}}\otimes\cdots\otimes\mathcal{G}^{m}_{\ell_{[m]}}\right)
\left(\mathcal{E}^{1}_{\ell_{1}}\otimes\cdots\otimes\mathcal{E}^{m}_{\ell_{m}}\right) \rho^{\otimes n}.
\end{align*}
A number $R$ is an achievable rate if for every $n$ there exist 
$(n,\varepsilon)$-protocols, with $\varepsilon\rightarrow0$ and $k/n\rightarrow R$ 
as $n\rightarrow\infty$. 
The non-interactive GHZ distillation capacity of $\rho$ is defined as
\begin{align}
  D_{\text{n.i.}}(\rho) \coloneqq \sup\{R:R \ \text{achievable} \}.
\end{align}
\end{definition}

\subsection{GHZ distillation via EPR state generation}
\label{merging} 

We start with describing two ``baseline'' protocols for the distillation of GHZ states, both of which proceed through first creating EPR pairs between certain designated pairs of players, and finally using teleportation to fuse them into GHZ states. 

\begin{theorem}
\label{thm:GHZ-combing}
    Let $\rho^{A_1\ldots A_m}$ be held by $m$ parties. The following rate of GHZ is distillable under LOCC: 
    \begin{align*}
       D_{\exists} = \max_{i\in[m]}\left\{ \min_{\emptyset\neq J\subseteq [m]\setminus\{i\}}{\frac{I(A_J\rangle A_{[m]\setminus J})_\rho}{\abs{J}}} \right\}.
    \end{align*}
\end{theorem}
\begin{proof}
The proof is founded on the entanglement combing protocol, where an initial pure entangled state is transformed into EPR pairs between a distinguished party $i$ (the ``root'') and the other parties $[m]\setminus\{i\}$ (the ``leaves'') \cite{PhysRevLett.103.220501}. The protocol is actually repeated state merging \cite{state-merging}, from the ``leaves'' to the ``root'', which is why we can apply it to a mixed state. 
In \cite{PhysRevLett.103.220501}, for each root node the complete rate region of the $m-1$ EPR rates between $j\in[m]\setminus\{i\}$ and $i$ is given. It is described as the convex hull of its extreme points, each of which is given by one of the different orders in which the leaf nodes are merged to the root; all the other points of the rate region are achieved by time sharing, i.e.~convex combination. We note in passing that using the new bounds of \cite[Sec.~6.3~{\&}~Cor.~6.12]{simultaneous-decoupling}, one can also understand it as all merging steps being done simultaneously, which allows one to attain the points of the rate region directly without going through time sharing of the extreme points.

Our present state is mixed, so we consider a purification $\psi^{A_1\ldots A_mE}$ of $\rho^{A_1\ldots A_m}$ and run a virtual combing protocol on the pure state, where of course $E$ is not actually participating, which is why that party has to go last in the iterative protocol of \cite{PhysRevLett.103.220501}.
That means that the $(m-1)!$ extreme points of the rate polytope in $\mathbbm{R}^{m-1}$ correspond to $m-1$ players merging their states in different orders to the root $i$, and in each case the environment $E$ merges its state at the end (this is because that last step is only performed virtually, as the environment does not actually contribute to the protocol). Using time sharing on the $(m-1)!$  different orders of merging $m-1$ parties to the $i$-th party, the region of attainable tuples of rates $D^{\text{EPR}}(i:j)$ of EPR states between player $i$ and $j\neq i$, results in 
\begin{align*}
    \forall J\subseteq [m]\setminus\{i\} \quad \sum_{j\in J} D^{\text{EPR}}(i:j) \leq I(A_{J}\rangle A_{[m]\setminus J})_\rho.
\end{align*}
To then create GHZ states from the combed entanglement we use simple teleportation from the root $i$ to all leaves $j\in[m]\setminus\{i\}$, and for this to work all EPR rates have to be equal, i.e.~$D^{\text{EPR}}(i:1) = \ldots = D^{\text{EPR}}(i:i-1) = D^{\text{EPR}}(i:i+1) = \ldots = D^{\text{EPR}}(i:m) =: D$,
so that our GHZ rate is achievable in this protocol if and only if
\begin{align*}
   \forall J\subseteq [m]\setminus\{i\} \quad \abs{J}D \leq I(A_{J}\rangle A_{[m]\setminus J})_\rho,
\end{align*}
which is maximized by $D = \min_{\emptyset\neq J\subseteq [m]\setminus\{i\}} \frac{I(A_{J}\rangle A_{[m]\setminus J})_\rho}{\abs{J}}$. Finally we optimize over the choice of the distinguished party, 
which concludes the proof.
\end{proof}

Next, we develop another protocol and associated rate for GHZ distillation that is based on the assisted distillable entanglement for mixed states. For two parties, $i$ and $j$, this is the largest rate of EPR pairs distillable by LOCC from $\rho^{\ox n}$, and denoted $E_A(i:j|\rho)$. After distilling EPR pairs between all adjacent nodes of a spannign tree of the $m$ parties we can fuse them together by teleportation to obtain GHZ states.
We start by recalling a lower bound on this quantity due to Dutil and Hayden \cite{Dutil-Hayden}, see also \cite[Cor.~6.13]{simultaneous-decoupling}.

\begin{definition}
    For a multipartite state $\rho^{A_1\ldots A_m}$, the min-cut coherent information between parties $A_i$ and $A_j$ is defined as follows:
    \begin{align*}
       I_{\text{min-cut}}\left(A_i\rangle A_j\right)_\rho 
        \coloneqq \min_{J\subseteq [m]\setminus \{i,j\}} I_{\text{c}}\left(A_iA_{J}\rangle A_{[m]\setminus(J\cup\{i\})}\right)_\rho.
    \end{align*}
\end{definition}

\begin{lemma}[{Cf.~Dutil~{\&}~Hayden~\cite[Thm.~14]{Dutil-Hayden}}]
Let $\rho^{A_1\ldots A_m}$ be an $m$-partite state. The asymptotic entanglement of assistance between parties $A_i$ and $A_j$ is lower bounded by
\begin{equation}\begin{split}
  &E_A(i:j|\rho) 
  \geq \\
  &\sup_{\mathcal{T}_j:A_j\rightarrow B_j \atop \text{cptp maps}} 
       \max\{ 
        I_{\text{min-cut}}(B_i\rangle B_j)_\sigma,
        I_{\text{min-cut}}(B_j\rangle B_i)_\sigma
        \}\\
  &\phantom{=======}
        \text{s.t. } 
        \sigma^{B_1\ldots B_m} = (\mathcal{T}_1\ox\cdots\ox\mathcal{T}_m)\rho^{A_1\ldots A_m}.
\end{split}\end{equation} 
\end{lemma}

\begin{theorem}
\label{thm:EoA-to-GHZ}
Let $\rho^{A_1\ldots A_m}$ be held by $m$ parties. The following rate of GHZ is distillable under LOCC:
\begin{align*}
    D_{\text{EoA}}(\rho) =
        \max_{G=([m],E) \atop \text{spanning tree}}
        \left( \sum_{ij\in E}
        \frac{1}{E_A(i:j|\rho)} \right)^{-1}.
\end{align*}
\end{theorem}
\begin{proof}
We use assisted entanglement distillation \cite{Dutil-Hayden,Dutil:PhD}, yielding rates $R_e = E_A(i:j)$ of EPR pairs between players $i$ and $j$, for the edges $e = ij\in E$ of a spanning tree $G=([m],E)$ on $m$ vertices. We apply this procedure on larger blocks of states for smaller $R_e$; the basic ingredient is time-sharing, as follows. Let $0\leq\lambda_e\leq 1$, $\sum_{e\in E}\lambda_e=1$. For $n$ initial states and the edge $e=ij\in E$, we use $\lambda_e n$ copies of the tensor product to distill entanglement between the parties $i$ and $j$ with the others helping by LOCC. For the whole block then, there are EPR pairs between $i$ and $j$ at asymptotic rate $\lambda_e R_e$. From these EPR pairs along the spanning tree $G$ we thus get an achievable rate of GHZ states $R := \min_e \lambda_e R_e$. To optimize this rate $R$, all the $\lambda_e R_e$ have to be equal, i.e.~$\lambda_e = \frac{R}{R_e}$. From the normalisation $\sum_e\lambda_e=1$ we finally obtain the result.
\end{proof}

\medskip
Notice that both Theorems \ref{thm:GHZ-combing} and \ref{thm:EoA-to-GHZ} rely on the idea that if you have EPR pairs along the edges of a connected graph of the $[m]$ notes of a network, then by teleportation a GHZ state can be constructed. Only that in the former result, due to the use of entanglement combing, the network is restricted to star-graphs anchored at an arbitrary node $i$; in the latter result restricting to star-graphs $G$ would potentially result in a lower rate. 
However, for a star-graph, the teleportation protocol can create an arbitrary $m$-qubit state, not only GHZ states.

\subsection{Genuine multipartite GHZ distillation}
\label{make-coherent}
In essence, the first concept behind making protocols coherent involves transforming classical symbols, represented by $x$, into basis states $\ket{x}$ within the Hilbert space. Functions $f : x \rightarrow f(x)$ then give rise to linear operators on the Hilbert space, with particular interest lying in permutations (or one-to-one functions) as they lead to unitaries (or isometries). The second notion revolves around achieving reversibility in classical computations by extending them into one-to-one functions. Lastly, we utilize local decoding operations, which are completely positive trace-preserving (cptp) maps represented by their isometric Stinespring dilations \cite{Stinespring}. In summary, making coherent allows us to replace
probabilistic mixtures by quantum superpositions, transforming a classical protocol working on individual letters into a set of unitaries. These unitaries then act as permutations on the basis states preserving coherent quantum superpositions. It is therefore conceivable that our secret key generating protocol could be converted into a (pure) entanglement generating protocol by executing all the steps coherently.

Another crucial element in our proof is the covering by constant type classes. To ensure self-containment, we provide a brief discussion of type classes and present the covering lemma.
For sequences of length $n$ from a finite alphabet $\mathcal{X}$, denoted generically as $x^n=x_1\ldots x_n\in\mathcal{X}^n$, we define the type of $x^n$ as the empirical distribution of letters in $x^n$. In other words, $p$ is the type of $x^n$ if
\begin{align*}
    \forall x\in\mathcal{X},\quad p(x) = \frac{1}{n}\abs{\{k:x_k=x\}}.
\end{align*}
The type class of $p$, denoted by $\mathcal{T}^n_p$, is defined as the set of all sequences of length $n$ with type $p$. Clearly, any type class can be obtained by considering all permutations of an arbitrary sequence with that type. The subsequent statement represents a basic characteristic of type classes:
\begin{align*}
    (1+n)^{-\abs{\mathcal{X}}}2^{nS(X)_p}\leq \abs{\mathcal{T}_p^n}\leq 2^{nS(X)_p},
\end{align*}
where $S(X)_p$ is the (Shannon) entropy of the random variable $X$.

\begin{lemma}[{Devetak~{\&}~Winter~\cite[Prop.~4]{2005RSPSA.461..207D}}]
\label{covering}
For a classical-quantum channel $G:\mathcal{X}\rightarrow B$ and a type $p$, let $U^{(j)}$ be i.i.d. according to the uniform distribution on the type class $\mathcal{T}^n_p$, $j=1,\ldots,M$. Define the state
\begin{align*}
       \sigma(p)=\frac{1}{\abs{\mathcal{T}^n_p}}\sum_{x^n\in\mathcal{T}_p^n}G^n_{x^n}
       =\mathbb{E}G^n_{U^{(j)}}.
\end{align*}
Then for every $\varepsilon,\delta>0$, and sufficiently large $n$, 
\begin{align*}
      \Pr&\left\{\left\|\frac{1}{M}\sum_{j=1}^{M}G^n_{U^{(j)}}-\sigma(p)\right\|_1\geq\varepsilon\right\} \\
       &\phantom{==========}
        \leq 2|B|^n\exp\left(-M\iota^n\frac{\varepsilon^2}{288\ln2}\right),
\end{align*}
where $\log \iota = -I(X;B)-\delta$. 
\end{lemma}

Now we are prepared to present our main result on the distillable GHZ states in the following theorem.

\begin{theorem}
\label{thm:GHZ-rate}
For any state $\rho^{A_1\ldots A_m} = \Tr_E \psi^{A_1\ldots A_m E}$, 
as purified to the environment, and pure instruments 
$\mathcal{E}^{(j)}=\{\mathcal{E}^{(j)}_{x_j}\}_{x_j\in\mathcal{X}_j}$, for 
each player $j\in[m]$, meaning that each $\mathcal{E}^{(j)}_{x_j}:A_j\rightarrow A'_j$ 
is a cp map of Kraus rank one, i.e. $\mathcal{E}^{(j)}_{x_j}(\sigma)=E^{(j)}_{x_j}\sigma (E^{(j)}_{x_j})^{\dagger}$. let 
\[
  \omega^{X_{[m]}A'_{[m]}E} 
    = \sum_{x_{[m]}} \proj{x_{[m]}}^{X_{[m]}} \ox (\mathcal{E}_{x_{[m]}}\ox\id_E)\psi,
\]
where $\mathcal{E}_{x_{[m]}} = \mathcal{E}^{(1)}_{x_1}\ox\cdots\ox\mathcal{E}^{(m)}_{x_m}$.
Then, for any $j\in[m]$,
\begin{align}
\label{GHZ-cq}
  D(\rho) \geq S(X_1\ldots X_m|EA'_{[m]\setminus j})-R_{\text{CO}}^{\text{cq}},
\end{align}
where $R_{\text{CO}}^{\text{cq}}$ is the rate of communication for omniscience
of $X_{[m]}$ from Lemma \ref{cq-omniscience}.
\end{theorem}
\begin{proof}
Let $\ket{\psi}^{A_1\ldots A_mE}$ be a purification of $\rho^{A_1\ldots A_m}$.
\begin{figure*}[t]
\begin{equation}\begin{split}
  \label{eq:psi-bar}
  \ket{\overline{\psi}} = \sum_{x_{[m]}^n} \sqrt{p^n(x_{[m]}^n)} \ket{\ell_{[m]}}^{\otimes m}
    &\left( \sum_{\forall i\, \xi_i^n\in f_i^{-1}(\ell_i)} 
               \sqrt{\Delta^{(1,\ell_{[m]})}_{\xi_{[m]}^n}} \otimes
                      \ket{\xi_{[m]}^{n}} \right) \\
    &\phantom{=}
    \otimes \cdots \\
    &\phantom{=}
     \otimes\left( \sum_{\forall i\, \xi_i^n\in f_i^{-1}(\ell_i)}
     \sqrt{\Delta^{(m,\ell_{[m]})}_{\xi_{[m]}^n}} \ket{x_{m}^{n}} \otimes \ket{\xi_{[m]}^{n}} \right) \ket{x^n_{[m]}}\ket{\widehat{\psi}_{x_{[m]}^n}}^{A_{m}^{\prime n}E^n}.
\end{split}\end{equation}
\hrulefill
\end{figure*}
The protocol starts by each player applying its instrument coherently on its system. Given the rank one Kraus operators $\{E_{x_j}^{j}\}_j$, the coherent instruments result in isometries $V_i:A_{i}\hookrightarrow A'_{i}\otimes X_{i}$ defined as $V_i = \sum_{x_{i}\in\mathcal{X}_{i}} E^{(i)}_{x_i} \otimes \ket{x_i}$. 
The isometries act as follows on a single copy:
\begin{align*}
\ket{\widehat{\psi}}&= (V_{1}\otimes\cdots\otimes V_{m}\otimes \1^{E})\ket{\psi}^{A_{[m]}E}\\
                    &= \sum_{x_{[m]}} (E^{(1)}_{x_{1}}\otimes\cdots\otimes E^{(m)}_{x_{m}}\otimes \1^{E})\ket{\psi}^{A_{[m]}E}
                    \otimes\ket{x_{[m]}}\\
                    &= \sum_{x_{[m]}} \sqrt{p(x_{[m]})}\ket{\widehat{\psi}_{x_{[m]}}}^{A'_{[m]}E}\otimes\ket{x_{[m]}},
\end{align*}
where $E^{([m])}_{x_{[m]}}=E^{(1)}_{x_{1}}\otimes\cdots\otimes E^{(m)}_{x_{m}}$ and 
\begin{align*}
p(x_{[m]}) 
  &= \bra{\psi}(E^{([m])}_{x_{[m]}}\otimes\1)^{\dagger}
  (E^{([m])}_{x_{[m]}}\otimes\1)\ket{\psi},\\
\ket{\widehat{\psi}_{x_{[m]}}}^{A'_{[m]}E}
  &= \frac{1}{\sqrt{p(x_{[m]})}}(E^{([m])}_{x_{[m]}}\otimes\1)\ket{\psi}^{A_{[m]}E}.
\end{align*}
The instruments are applied independently on each copy, therefore with $n$ copies of the initial pure state, we want to distill GHZ states from $n$ copies of $\ket{\widehat{\psi}}$:
\begin{align*}
  \ket{\widehat{\psi}}^{\otimes n}
    = \sum_{x_{[m]}^n} \sqrt{p^n(x_{[m]}^n)} \ket{x_{1}^n}\cdots\ket{x_{m}^n}
    \otimes\ket{\widehat{\psi}_{x_{[m]}^n}}^{A_{[m]}^{\prime n}E^n} \!\!\!,
\end{align*}
where systems $A_{[m]}^{\prime n}$ in $\ket{\widehat{\psi}_{x_{[m]}^n}}^{A_{[m]}^{\prime n}E^n}$ are the quantum side information
at the disposal of the players to help them with their respective decoding, and system $E$ is the eavesdropper's quantum side information.

The next step is to achieve omniscience, where each player coherently computes its hash value
and broadcasts it coherently to the other players via teleportation through GHZ states. In detail,
let $f_{j}:\mathcal{X}_j^n \rightarrow \mathcal{L}_j$ be the Slepian-Wolf 
hash function used by party $j$ in the classical 
part of the protocol of Lemma \ref{cq-omniscience} (omniscience), 
and $\left(\Delta^{(j,\ell_{[m]})}_{x_{[m]}^n}:x_{[m]}^n\right)$ the POVM (decision rule) 
that they use to recover $x_{[m]}^n$ when the classical messages 
$\ell_{[m]}$ are broadcast. 
Each party $j$ will apply an isometry based on the mappings 
$x_{j}^{n} \longmapsto (f_{j}(x_{j}^{n}),x_{j}^{n})$ for $j\in[m]$, namely
\begin{align*}
  W_j = \sum_{x_{j}^{n}} \ket{f_{j}(x_{j}^{n}),x_{j}^{n}}\!\bra{x_{j}^{n}},
\end{align*}
where $\ket{\ell}=\ket{f_{j}(x_{j}^{n})}$ are computational basis for some Hilbert 
space $U_j = \text{span}\{\ket{\ell} : \ell\in \mathcal{L}_j\}$.

The state at the end of this step is 
\begin{align*}
  &\ket{\psi'} = \sum_{x_{[m]}^n} \sqrt{p^n(x_{[m]}^n)}
                 \ket{x_{1}^{n},f_{1}(x_{1}^{n})}\cdots\ket{x_{m}^{n},f_{m}(x_{m}^{n})}\\
  &\phantom{====================}
   \otimes \ket{\widehat{\psi}_{x_{[m]}^n}}^{A_{[m]}^{\prime n}E^n}.
\end{align*}
Next, the coherent transmission of the hash value $\ell_j$ to other parties follows, effectively implementing a multi-receiver cobit channel.
\cite{PhysRevLett.92.097902}, 
i.e. party $j$ aims to implement the isometry
$\ket{\ell_j} \longmapsto \ket{\ell_j}^{\otimes m}$. 
This multi-receiver cobit channel can be implemented by utilizing GHZ states for teleportation. In order to coherently transmit $nR_{j}$ bits, 
where $R_{j}\coloneqq \frac{1}{n}\log\abs{\mathcal{L}_j}$, 
$nR_{j}$ GHZ states are needed, i.e. the following state:
\begin{align*}
  \ket{\Gamma_{m}}^{\otimes nR_{j}}
    =\left(\frac{1}{\sqrt{2}}(\ket{0}^{\otimes m}+\ket{1}^{\otimes m})\right)^{\otimes nR_{j}}.
\end{align*}
After implementing the multi-receiver cobit channel, the $j$-th party possesses its initial share $\ket{x_{j}^{n}}$, along with all the hash values broadcast to it. Consequently, the overall state is as follows:
\begin{align*}
  \ket{\widetilde{\psi}} 
    &=\sum_{x_{[m]}^n} \sqrt{p^n(x_{[m]}^n)} 
        \ket{x_{1}^{n},f_{1}(x_{1}^{n})\ldots f_{m}(x_{m}^{n})}\otimes \cdots\\
        &\phantom{=======}
         \otimes \ket{x_{m}^{n},f_{1}(x_{1}^{n})\ldots f_{m}(x_{m}^{n})} 
        \otimes \ket{\widehat{\psi}_{x_{[m]}^n}}^{A_{[m]}^{\prime n}E^n}
\end{align*} 
Having received the hash values, the parties then proceeds to recover $x_{[m]}^n$ locally. Each party independently runs its Slepian-Wolf decoder coherently to deduce the $\ket{x^{n}_{j}}$ values of the other $m-1$ parties. Specifically, the $j$-th party applies the following controlled isometry on its corresponding systems:
\begin{align*}
 \sum_{\ell_{[m]}} \ketbra{\ell_{[m]}} \otimes W_{D}^{(j,\ell_{[m]})},
\end{align*}
where the coherent measurement isometry of the $j$-th party is defined as:
\begin{align}
 W_{D}^{(j,\ell_{[m]})}
   &=\sum_{\forall i\in [m]\, \xi_i^n\in f_i^{-1}(\ell_i)} \sqrt{\Delta^{(j,\ell_{[m]})}_{\xi_{[m]}^n}} \otimes \ket{\xi_{[m]}^n},
\end{align}
with $\Delta^{(j,\ell_{[m]})}_{\xi_{[m]}^n}$ the POVM elements of the $j$-th decoder, which acts on $X_j^n{A'{}}^{n}_j$.
The classical-quantum omniscience result of \cite{FSAW:GHZ} guarantees successful decoding if the rates $R_{[m]}$ fulfill the conditions of Lemma \ref{cq-omniscience}. 
The state after each party has applied their decoding isometry is $\ket{\overline\psi}$ displayed in Eq.~\eqref{eq:psi-bar} at the top of the page,
where we have utilized the notation $\ell_{[m]}=f_{1}(x_{1}^{n})\ldots f_{m}(x_{m}^{n})$.
The coherent gentle measurement lemma \cite{796385,4544968} ensures that for each party $j\in [m]$, if the decoding error is no greater than $\varepsilon_1$ (which is guaranteed by Lemma \ref{cq-omniscience}), the following two states are $2\sqrt{\varepsilon_1(2-\varepsilon_1)}$-close in trace distance:
\begin{align*}
&\sum_{x^n_{[m]}} \sqrt{p(x^n_{[m]})} \ket{\ell_{[m]}}\ket{x^n_{[m]\setminus j}} \\
&\phantom{===:} \otimes
 \sum_{\forall i\, \xi_i^n\in f_i^{-1}(\ell_i)} \sqrt{\Delta^{(j,\ell_{[m]})}_{\xi_{[m]}^n}}\ket{x_{j}^{n}} \ket{\widehat{\psi}_{x_{[m]}^n}}^{A_{[m]}^{\prime n}E^n}
\otimes\ket{\xi_{[m]}^{n}} \!,
\end{align*}
and
\begin{align*}
\sum_{x_{[m]}^n} \sqrt{p(x_{[m]}^n)} \ket{\ell_{[m]}} \ket{x_{[m]}^{n}} 
\ket{\widehat{\psi}_{x_{[m]}^n}}^{A_{[m]}^{\prime n}E^n} \otimes \ket{x_{[m]}^{n}}.
\end{align*}
By using triangle inequality for the trace distance $m$ times, the state $\ket{\overline{\psi}}$ will be $2m\sqrt{\varepsilon_1(2-\varepsilon_1)}$-close in trace distance to the following state:
\begin{align*}
  \ket{\widehat{\psi}}
   &= \sum_{x_{[m]}^n} \sqrt{p(x^{n}_{[m]})}\ket{x_{[m]}^n,\ell_{[m]}} \cdots \ket{x_{[m]}^n,\ell_{[m]}} \ket{\widehat{\psi}_{x_{[m]}^n}}^{A_{[m]}^{\prime n}E^n}\!\!.
\end{align*}
All parties proceed to clean up their $L_{[m]}$-registers through the application of local unitaries. Specifically, they extend the isometries $W_j: \ket{x^{n}_{j}}\ket{0}^{E}\mapsto \ket{x^{n}_{j}}\ket{f_{j}(x^{n}_{j})}$ to unitaries by defining $\ket{x^{n}_{j}}\ket{i}^{E}\mapsto \ket{x^{n}_{j}}\ket{i+f_{j}(x^{n}_{j})}$, where the addition is performed within an abelian group on the ancillary register (e.g. integers modulo $|\mathcal{L}_j|$). The state will be transformed into
\begin{align*}
     \ket{\widetilde{\psi}}
    &= \sum_{x_{[m]}^n} \sqrt{p(x^{n}_{[m]})}\ket{x_{[m]}^n} \cdots \ket{x_{[m]}^n}  \ket{\widehat{\psi}_{x_{[m]}^n}}^{A_{[m]}^{\prime n}E^n},
\end{align*}
with residual states $\ket{\widehat{\psi}_{x_{[m]}^n}}$ on ${A_{[m]}^{\prime n}}$ and $E^n$. 
One of the players measures the joint type $q$ non-destructively and informs the other players
about the result.
The protocols aborts if $q$ is not typical, i.e. if $\|p-q\|_1 \geq \delta$.

This leaves the players sharing the post-measurement state
\begin{align}
\label{omni-reached}
\frac{1}{\sqrt{|\mathcal{T}_{q}^n}|}
\sum_{x_{[m]}^n\in\mathcal{T}_{q}^n}\ket{x_{[m]}^n}\cdots\ket{x_{[m]}^n}\otimes\ket{\widehat{\psi}_{x_{[m]}^n}}^{A_{[m]}^{\prime n}E^n}.
\end{align}
We now consider, for each $j\in [m],$ a random partition of the type 
 class $\mathcal{T}_q^n $ into $|\mathcal{K}|$ blocks of size $|\mathcal{S}|$, where
\begin{align*}
|\mathcal{K}||\mathcal{S}|&=|\mathcal{T}_{q}^n|\approx 2^{nH(X_{[m]})},\\
|\mathcal{S}|&\approx2^{nI(X_{[m]};EA'_{[m]\setminus j})}.
\end{align*}
Define an isometry
$
\sum_{x_{[m]}^n\in\mathcal{T}_{q}^n}\ket{k(x_{[m]}^n),s(x_{[m]}^n)}\bra{x_{[m]}^n},
$
where $k:\mathcal{T}_q^n\rightarrow
\mathcal{K}$ labels the block and $s:\mathcal{T}_q^n\rightarrow
\mathcal{S}$ labels the elements within each block, such that there is a one-to-one correspondence between
$(k,s)$ and $x_{[m]}^n(k,s)$. 
All players apply this unitary locally, evolving the state to 
\begin{align*}
\frac{1}{\sqrt{|\mathcal{K}||\mathcal{S}|}}
\sum_{k,s}\ket{k,s}\cdots\ket{k,s}\otimes\ket{\widehat{\psi}_{x_{[m]}^n(k,s)}}^{A_{[m]}^{\prime n}E^n}.
\end{align*}
All players but the distinguished player $j$
measure the $s$-component of their registers in the Fourier
conjugate basis:
\begin{align*}
\left\{\ket{\widehat{t}}=
\frac{1}{\sqrt{|\mathcal{S}|}}\sum_{s=1}^{|\mathcal{S}|}e^{2\pi i st\slash|\mathcal{S}|}\ket{s}:t=1,\ldots,|\mathcal{S}|\right\},
\end{align*}
and inform player $j$ about their results $t_{[m]\setminus j}$, who in turn
applies the phase shift operator
\begin{align*}
\sum_{s=1}^{|\mathcal{S}|}e^{-2\pi i s\slash|\mathcal{S}| \sum_{z\in[m]\setminus j}t_z}\ketbra{s}.
\end{align*}
So far we obtained the following state:
\begin{align*}
\frac{1}{\sqrt{|\mathcal{K}||\mathcal{S}|}}
\sum_{k,s}\ket{k}\cdots\ket{k,s}\cdots\ket{k}\otimes\ket{\widehat{\psi}_{x_{[m]}^n(k,s)}}^{A_{[m]}^{\prime n}E^n}.
\end{align*}
Absorbing $s$-component of player $j$ into $A^{\prime n}_j$, the above state can be written as:
\begin{align}
\label{raw-ghz}
\ket{\Theta}=\frac{1}{\sqrt{|\mathcal{K}|}}
\sum_{k}\ket{k}\cdots\ket{k}\otimes\ket{\theta_{x_{[m]}^n(k,s)}},
\end{align}
where
\begin{align*}
\ket{\theta_{x_{[m]}^n(k,s)}}=
\frac{1}{\sqrt{|\mathcal{S}|}}\sum_{s=1}^{|\mathcal{S}|}\ket{s}\otimes\ket{\widehat{\psi}_{x_{[m]}^n(k,s)}}^{A_{[m]}^{\prime n}E^n}.
\end{align*}
The reduced states on $A^{\prime n}_{[m]\setminus j}E^n$ of 
$\theta_{x_{[m]}^n(k,s)}$
(for each $k$) is
\begin{align*}
\Tr_{A^{\prime n}_j}\theta_{x_{[m]}^n(k,s)}=
\nu_{k}^{A^{\prime n}_{[m]\setminus j}E^n}=
\frac{1}{|\mathcal{S}|}\sum_{s=1}^{|\mathcal{S}|}\widehat{\psi}_{x_{[m]}^n(k,s)}^{A^{\prime n}_{[m]\setminus j}E^n}.
\end{align*}
According to the constant type covering given by Lemma \ref{covering}, 
for all $k$, if $\log\mathcal{|S|}\geq n(I(X_{[m]};A'_{[m]\setminus j}E)+\delta)$, then
\begin{align*}
\frac{1}{2}\left\|\nu_{k}^{A^{\prime n}_{[m]\setminus j}E^n}-\sigma(q)\right\|_{1}\leq\varepsilon_2,
\end{align*}
where 
\begin{align*}
\sigma(q)=
\frac{1}{|\mathcal{T}_{q}^n|}\sum_{x_{[m]}^n\in\mathcal{T}_{q}^n}\Tr_{A^{\prime n}_j}\theta_{x_{[m]}^n}^{A^{\prime n}_{[m]}E^n},
\end{align*}
in which $\Tr_{A^{\prime n}_j}\theta_{x_{[m]}^n}^{A^{\prime n}_{[m]}E^n}$ is understood from Eq.~\eqref{omni-reached}. 
From the relation between trace distance and fidelity,
we obtain $F(\nu_{k}^{A^{\prime n}_{[m]\setminus j}E^n},\sigma(q))\geq 1-\varepsilon_2$.
Let $\ket{\zeta}^{RA'_{[m]\setminus j}E}$ be a purification of $\sigma(q)$ with purifying system $R$.
Since mixed-state fidelity equals the maximum pure-state fidelity over all purifications
of the mixed states, and all purifications are related by unitaries on the purifying systems, 
there are unitaries $U_k$ for player $j$, one for each $k\in\mathcal{K}$, such that 
\begin{align*}
F\left((U_k\otimes\1^{A^{\prime n}_{[m]\setminus j}E^n})\ket{\theta_{x_{[m]}^n(k,s)}},\ket{\zeta}\right)\geq 1-\varepsilon_2.
\end{align*}
This means that if $j$ applies $\sum_{k}\ketbra{k}\otimes U_k$ to its share of $\ket{\Theta}$ in Eq.~\eqref{raw-ghz}, this state is transformed into a state $\ket{\Theta'}$ such that
\begin{align*}
F\left(\ket{\Theta'},\frac{1}{\sqrt{|\mathcal{K}|}}\sum_{k}\ket{k}\cdots\ket{k}\otimes\ket{\zeta}\right)\geq 1-\varepsilon_2.
\end{align*}
Non-typical type happens with vanishing probability; in the event of typical type class, $m$ players
distill GHZ state of rate
\begin{align*}
H(X_{[m]})-I(X_{[m]};A'_{[m]\setminus j}E)-R_{\text{CO}}^{cq}.
\end{align*}
This concludes the proof.
\end{proof}

The following Theorem is a special case of Theorem \ref{thm:GHZ-rate} where players use full local measurements consisting of rank-one operators and communication.
\begin{theorem}
\label{thm:GHZ-rate-c}
For any state $\rho^{A_1\ldots A_m} = \Tr_E \psi^{A_1\ldots A_m E}$, 
as purified to the environment, with the notation of Theorem \ref{secret-key c} and where all local measurements are assumed to consist of rank-one POVM elements,
\begin{align}
\label{GHZ-c}
  D_{\text{n.i.}}(\rho) \geq S(X_1\ldots X_m|E)-R_{\text{CO}}^{\text{c}}.
\end{align}
\end{theorem}

\medskip
Comparing the rate \eqref{GHZ-cq} with the analogous one from \cite{GHZ:ISIT2020,FSAW:GHZ}, we are disappointed to see that we should have to condition on all but one of the $A'$-registers. Why was that not needed in the pure state case? The 
reason is that there we could decouple the block of $n$ systems $A_{[m]}^{\prime n}$
completely by first measuring the type class of $X_{[m]}^n$, which after the 
omniscience phase every player can do, and applying the same controlled permutation 
of the $A_j^{\prime n}$, controlled by the $j$-th player's copy of $X_{[m]}^n$, 
transforming $\psi_{x_{[m]}^n}^{A_{[m]}^{\prime n}}$ to a standard state that 
depends only on the type class of $x_{[m]}^n$. If there is correlation 
with $E^n$, this does not work, unless one would perform the same permutation 
on the eavesdropper's systems, which however are unaccessible to the legal 
players. 
Still, taking this possibility into account, 
we can achieve the following potentially improved rate compared to Theorem \ref{thm:GHZ-rate}:

\begin{theorem}
\label{thm:better-GHZ-rate} 
Under the same assumptions as Theorem \ref{thm:GHZ-rate}, and for $n$ 
i.i.d. repetitions, consider unitaries $U^{(i)}_{x_{[m]}^n}$ on $A_i^{\prime n}$,
for all $i=1,\ldots,m$ (for instance permutations of the $n$ systems). 
Then, for the state 
\begin{equation*}\begin{split}
  &\widetilde{\omega}^{X_{[m]}^nA_{[m]}^{\prime n}E^n} \\
  &\phantom{=}
   = \sum_{x_{[m]}^n} \proj{x_{[m]}^n}^{X_{[m]}^n}
         \ox \left(U^{(1)}_{x_{[m]}^n}\ox\cdots\ox U^{(m)}_{x_{[m]}^n}\ox\1\right) \\
  &\phantom{=================}      \bigl((\mathcal{E}_{x_{[m]}^n}\ox\id_{E^n})\psi^{\ox n}\bigr) \\
  &\phantom{===============:}
  \left(U^{(1)}_{x_{[m]}^n}\ox\cdots\ox U^{(m)}_{x_{[m]}^n}\ox\1\right)^\dagger,
\end{split}\end{equation*}
and any $j\in[m]$,
\begin{align*}
  \phantom{====}
  D_{\text{n.i.}}(\rho) 
    \geq \frac1n S(X_{[m]}^n|E^nA_{[m]\setminus j}^{\prime n})_{\widetilde{\omega}} - R_{\text{CO}}^{\text{cq}}.
  \phantom{=====}
  \blacksquare
\end{align*}
\end{theorem}

\medskip
Regarding the question of optimality of the GHZ distillation rates in Theorems \ref{thm:GHZ-rate} and \ref{thm:better-GHZ-rate}, we can elucidate it by considering a pure state coherent version of the example (\ref{eq:against-CO}):
\begin{equation}\begin{split}
\label{eq:against-CO-pure}
&\ket{\varphi}^{ABC} 
  = \frac{1}{\sqrt{dk^3}} \sum_{x=1}^d \sum_{\alpha,\beta,\gamma=1}^k \left(U_\alpha\ket{x}\ox\ket{\beta}\right)^{A} \\
  &\phantom{==========}
   \ox \left(V_\beta\ket{x}\ox\ket{\gamma}\right)^{B} \ox \left(W_\gamma\ket{x}\ox\ket{\alpha}\right)^{C}.
\end{split}\end{equation}
As before, it is evident that by a simple non-interactive communication protocol we can obtain a rate of $\log d$ GHZ states from this: every party measures $\alpha$, $\beta$ and $\gamma$, respectively, and broadcasts the value found to the others, who apply the appropriate local unitary $U_\alpha^\dagger$, $V_\beta^\dagger$ and $W_\gamma^\dagger$, respectively.

However, our present protocols are based on the GHZ correlation coming out of omniscience regarding a value obtained before the first communication, by identifying a local basis. Similar to the reasoning in Subsection \ref{subsec:n-i-limits}, it follows that the maximum correlation between one party and the other two, and hence any GHZ rate, is upper bounded by $1+\frac12\log d$ in the case of $k=2$ and the unitaries $\1$ and the quantum Fourier transform.

\section{Discussion}
The present results for secret key distillation with eavesdropper neatly 
generalise our earlier ones without eavesdropper \cite{FSAW:GHZ,GHZ:ISIT2020}, 
which are indeed 
recovered for the case of an initial product state $\rho^{A_1\ldots A_m}\ox\rho^E$,
in particular a trivial $E$-system. The difference is merely that the 
entropy $H(X_{[m]})$ after the omniscience protocol is replaced by
the conditional entropy $S(X_{[m]}|E)$, which makes sense as we need to sacrifice additional rate due to privacy amplification. 

Furthermore, we note that in contrast to the case of an initial pure state, discussed in \cite{GHZ:ISIT2020,FSAW:GHZ}, the attainable 
GHZ rate is smaller than the secret key rate, due to the difficulty of 
making the key distillation protocol coherent; concretely, the last step 
of disentangling (de-correlating) certain registers of the $m$ players from the GHZ state. 
Observe that this is an issue only for entanglement distillation compared
to the generation of secret key: for the latter, it is of no concern, and 
indeed can happen easily due to the employed protocol, that the secret key 
shared is correlated with other registers of the legal users generated 
during the protocol. 

To elucidate this further, consider the most general state of a 
perfect secret key between $m$ players and an eavesdropper, keeping track 
of all available information generated in a prior distillation protocol. 
This must be a pure $(m+1)$-partite state, where now each legal player 
has two registers, one for the key and one for residual quantum 
degrees of freedom (the ``shield''):
\begin{align*}
    \hspace{-.1cm}\ket{\psi}^{X_1B_1\ldots X_mB_mE} 
     = \frac{1}{\sqrt{d}}\sum_{x=1}^d \ket{x\ldots x}^{X_{[m]}}\ox \ket{\psi_x}^{B_1\ldots B_mE}, 
\end{align*}
where the $\ket{\psi_x}^{B_1\ldots B_mE}$ have the property that their 
reduced states $\psi_x^{E}$ on $E$ are all identical: $\psi_x^{E} = \sigma^E$ for all $x$. 
By the uniqueness of purifications up to local unitaries, this means that 
$\ket{\psi_x}^{B_1\ldots B_mE} = (U_x^{B_1\ldots B_m}\ox\1)\ket{\psi_0}^{B_1\ldots B_mE}$.
This form of state is known as \emph{pbit} \cite{pbit-PRL,pbit-IEEE,AH:mE}. 
It is well-known from these works that such a state, while 
containing $\log d$ bits of perfect secret key, can have arbitrarily small 
distillable entanglement; in fact, by compromising the quality of the key 
ever so slightly, the state after tracing out $E$ can be made to have positive
partial transpose (PPT) and hence be completely undistillable for entanglement. 
Thus, while it may seem that the $X$-registers do contain some kind of GHZ 
state, it is unavoidably decohered by the correlation with the $B$-registers, 
and in general there is no local way of undoing $U_x$.

\section*{Acknowledgment}
The authors thank Karla Gerstmann for a crucial discussion in the early stages of this project about the cost of compromised keys.

\bibliographystyle{ieeetr}
\bibliography{new-ghz.bib}

\begin{thebibliography}{10}

\bibitem{FSAW:GHZ-mixed-ISIT}
F.~Salek and A.~Winter, ``{Distillation of Secret Key and GHZ States From
  Multipartite Mixed States},'' in {\em Proc. 2022 IEEE Int'l Symp. Inf. Theory
  (ISIT), June 2022, Aalto, Finland}, pp.~2673--2678, IEEE, 2022.

\bibitem{1362905}
I.~Devetak and A.~Winter, ``Distilling common randomness from bipartite quantum
  states,'' {\em IEEE Trans. Inf. Theory}, vol.~50, pp.~3183--3196, December
  2004.

\bibitem{PhysRevA.72.052317}
J.~A. Smolin, F.~Verstraete, and A.~Winter, ``Entanglement of assistance and
  multipartite state distillation,'' {\em Phys. Rev. A}, vol.~72, p.~052317,
  November 2005.

\bibitem{PhysRevLett.93.230504}
I.~Devetak, A.~W. Harrow, and A.~Winter, ``{A Family of Quantum Protocols},''
  {\em Phys. Rev. Lett.}, vol.~93, p.~4, December 2004.

\bibitem{Maurer:key}
U.~M. Maurer, ``{Secret Key Agreement by Public Discussion from Common
  Information},'' {\em IEEE Trans. Inf. Theory}, vol.~39, no.~3, pp.~733--742,
  1993.

\bibitem{243431}
R.~Ahlswede and I.~Csisz\'ar, ``{Common randomness in information theory and
  cryptography. I. Secret sharing},'' {\em IEEE Trans. Inf. Theory}, vol.~39,
  pp.~1121--1132, July 1993.

\bibitem{Das:network}
S.~Das, S.~B\"auml, M.~Winczewski, and K.~Horodecki, ``Universal limitations on
  quantum key distribution over a network,'' {\em Phys. Rev. X}, vol.~11,
  p.~041016, 2021.

\bibitem{AH:mE}
R.~Augusiak and P.~Horodecki, ``Multipartite secret key distillation and bound
  entanglement,'' {\em Phys. Rev. A}, vol.~80, p.~042307, 2009.

\bibitem{MGKB}
G.~Murta, F.~Grasselli, H.~Kampermann, and D.~Bruss, ``{Quantum Conference Key
  Agreement: A Review},'' {\em Adv. Quantum Tech.}, vol.~3, p.~20000, 2020.

\bibitem{Streltsov-et-al}
A.~Streltsov, C.~Meignant, and J.~Eisert, ``{Rates of Multipartite Entanglement
  Transformations},'' {\em Phys. Rev. Lett.}, vol.~125, p.~080502, 2020.

\bibitem{Csiszar-Narayan.2004}
I.~Csisz\'ar and P.~Narayan, ``Secrecy capacities for multiple terminals,''
  {\em IEEE Trans. Inf. Theory}, vol.~50, pp.~3047--3061, December 2004.

\bibitem{PhysRevA.53.2046}
C.~H. Bennett, H.~J. Bernstein, S.~Popescu, and B.~Schumacher, ``Concentrating
  partial entanglement by local operations,'' {\em Phys. Rev. A}, vol.~53,
  pp.~2046--2052, April 1996.

\bibitem{BFG}
S.~Bravyi, D.~Fattal, and D.~Gottesman, ``{GHZ} extraction yield for
  multipartite stabilizer states,'' {\em J. Math. Phys.}, vol.~47, p.~062106,
  2006.

\bibitem{RevModPhys.81.865}
R.~Horodecki, P.~Horodecki, M.~Horodecki, and K.~Horodecki, ``Quantum
  entanglement,'' {\em Rev. Mod. Phys.}, vol.~81, pp.~865--942, Jun 2009.

\bibitem{2005RSPSA.461..207D}
I.~Devetak and A.~Winter, ``Distillation of secret key and entanglement from
  quantum states,'' {\em Proc. Roy. Soc. London Ser. A}, vol.~461,
  pp.~207--235, January 2005.

\bibitem{1998quant.ph..3033D}
D.~P. {DiVincenzo}, C.~A. {Fuchs}, H.~{Mabuchi}, J.~A. {Smolin}, A.~V.
  {Thapliyal}, and A.~{Uhlmann}, ``{Entanglement of Assistance},'' in {\em
  Proc. NASA Int'l Conf. Quantum Computing and Quantum Communications (QCQC
  1998)} (C.~P.~W. et~al., ed.), vol.~1509, p.~247–257, Springer Verlag, New
  York Ushuaia, 1998.
\newblock arXiv:quant-ph/9803033.

\bibitem{state-merging}
M.~Horodecki, J.~Oppenheim, and A.~Winter, ``{Quantum State Merging and
  Negative Information},'' {\em Commun. Math. Phys.}, vol.~269, no.~1,
  pp.~107--136, 2007.

\bibitem{Dutil-Hayden}
N.~Dutil and P.~Hayden, ``{Assisted Entanglement Distillation},'' {\em Quantum
  Inf. Comput.}, vol.~11, no.~5{\&}6, pp.~0496--0520, 2011.
\newblock arXiv[quant-ph]:1011.1972.

\bibitem{PhysRevLett.103.220501}
D.~Yang and J.~Eisert, ``{Entanglement Combing},'' {\em Phys. Rev. Lett.},
  vol.~103, p.~22050, November 2009.

\bibitem{Dutil:PhD}
N.~Dutil, {\em {Multiparty quantum protocols for assisted entanglement
  distillation}}.
\newblock PhD thesis, McGill University, Department of Computer Science, 2005.
\newblock arXiv[quant-ph]:1105.4657.

\bibitem{WINTER:PHD}
A.~Winter, {\em {Coding Theorems of Quantum Information Theory}}.
\newblock PhD thesis, Universit\"at Bielefeld, Department of Mathematics, July
  1999.
\newblock arXiv:quant-ph/9907077.

\bibitem{GHZ:ISIT2020}
F.~Salek and A.~Winter, ``{Multi-User Distillation of Common Randomness and
  Entanglement from Quantum States},'' in {\em Proc. 2020 IEEE Int'l Symp. Inf.
  Theory (ISIT), June 2020, Los Angeles, CA}, pp.~1967--1972, IEEE, 2020.

\bibitem{FSAW:GHZ}
F.~Salek and A.~Winter, ``{Multi-User Distillation of Common Randomness and
  Entanglement from Quantum States},'' {\em IEEE Trans. Inf. Theory}, vol.~68,
  pp.~976--988, February 2022.

\bibitem{renner:phd}
R.~Renner, {\em {Security of Quantum Key Distribution}}.
\newblock PhD thesis, ETH Z\"urich, Department of Physics, 2005.
\newblock arXiv:quant-ph/0512258.

\bibitem{Wilmink2003}
R.~Wilmink, {\em {Quantum Broadcast Channels and Cryptographic Applications for
  Separable States}}.
\newblock PhD thesis, University of Bielefeld, Department of Mathematics, 2003.
\newblock URN: nbn:de:hbz:361-3833.

\bibitem{pbit-PRL}
K.~Horodecki, M.~Horodecki, P.~Horodecki, and J.~Oppenheim, ``{Secure Key from
  Bound Entanglement},'' {\em Phys. Rev. Lett.}, vol.~94, p.~016050, April
  2005.

\bibitem{pbit-IEEE}
K.~Horodecki, M.~Horodecki, P.~Horodecki, and J.~Oppenheim, ``{General Paradigm
  for Distilling Classical Key From Quantum States},'' {\em IEEE Trans. Inf.
  Theory}, vol.~55, pp.~1898--1929, April 2009.

\bibitem{privacy-amplification:Bennett}
C.~H. Bennett, G.~Brassard, C.~Crepeau, and U.~M. Maurer, ``{Generalized
  Privacy Amplification},'' {\em IEEE Trans. Inf. Theory}, vol.~41,
  pp.~1915--1923, November 1995.

\bibitem{Tomamichel:book}
M.~Tomamichel, {\em Quantum Information Processing with Finite Resources --
  Mathematical Foundations}, vol.~5 of {\em SpringerBriefs in Mathematical
  Physics}.
\newblock Springer Verlag, New York Singapore Adelaide, 2016.

\bibitem{privacy-amplification:RK}
R.~Renner and R.~K\"onig, ``{Universally Composable Privacy Amplification
  Against Quantum Adversaries},'' in {\em Proc. 2005 Theory Crypto. Conf.
  (TCC)}, vol.~3378 of {\em LNCS}, pp.~407--425, Springer Verlag, Berlin
  Heidelberg, 2005.

\bibitem{MaassenUffink}
H.~Maassen and J.~B.~M. Uffink, ``{Generalized Entropic Uncertainty
  Relations},'' {\em Phys. Rev. Lett.}, vol.~60, pp.~1103--1106, March 1988.

\bibitem{simultaneous-decoupling}
P.~Colomer and A.~Winter, ``Decoupling by local random unitaries without
  simultaneous smoothing, and applications to multi-user quantum information
  tasks.'' arXiv[quant-ph]:2304.12114, 2023.

\bibitem{Stinespring}
W.~F. Stinespring, ``{Positive Functions on $C^*$-Algebras},'' {\em Proc. Amer.
  Math. Soc.}, vol.~6, no.~2, pp.~211--216, 1955.

\bibitem{PhysRevLett.92.097902}
A.~W. Harrow, ``{Coherent Communication of Classical Messages},'' {\em Phys.
  Rev. Lett.}, vol.~92, p.~097902, March 2004.

\bibitem{796385}
A.~{Winter}, ``Coding theorem and strong converse for quantum channels,'' {\em
  IEEE Trans. Inf. Theory}, vol.~45, pp.~2481--2485, Nov 1999.

\bibitem{4544968}
M.-H. {Hsieh}, I.~{Devetak}, and A.~{Winter}, ``{Entanglement-Assisted Capacity
  of Quantum Multiple-Access Channels},'' {\em IEEE Trans. Inf. Theory},
  vol.~54, pp.~3078--3090, July 2008.

\end{thebibliography}

\begin{IEEEbiographynophoto}{Farzin Salek}
is a Walter Benjamin Fellow with the Department of Mathematics at the Technical University of Munich, Germany. He is currently a visiting researcher at the Perimeter Institute for Theoretical Physics, Waterloo, Canada. Previously, he was a Ph.D. student with the Quantum Information Group (GIQ) at the Universitat Aut\`onoma de Barcelona, Spain, and with the Department of Signal Theory and Communications at the Universitat Polit\`ecnica de Catalunya, Spain. He received his Ph.D. degree (excellent cum laude) in December 2020.  
He has been awarded the Walter Benjamin Postdoctoral Fellowship by the German Research Foundation (DFG) and a Global Marie Sk{\l}odowska-Curie (MSCA) Fellowship to conduct research at Stanford University.
\end{IEEEbiographynophoto}

\begin{IEEEbiographynophoto}{Andreas Winter}
received a Diploma degree in Mathematics from Freie Universit\"at Berlin, Germany, in 1997, and a Ph.D. degree from Fakult\"at f\"ur Mathematik, Universit\"at Bielefeld, Germany, in 1999.
He was Research Associate at the University of Bielefeld until 2001, and then with the Department of Computer Science at the University of Bristol, UK. In 2003, still with the University of Bristol, he was appointed Lecturer in Mathematics, and in 2006 Professor of Physics of Information. From 2007 to 2012 he was in addition a Visiting Research Professor with the Centre of Quantum Technologies at NUS, Singapore. 
Since 2012 he has been ICREA Research Professor with the Universitat Aut\`onoma de Barcelona, Spain. His research interests include quantum and classical Shannon theory, and discrete mathematics.

He is recipient, along with Charles H. Bennett, Igor Devetak, Aram W. Harrow and Peter W. Shor, of the 2017 Information Theory Society Paper Award. In 2022, he received an Alexander von Humboldt Research Prize, a Hans Fischer Senior Fellowship of Technische Universit\"at M\"unchen, and one of three 2022 QCMC International Quantum Awards.
\end{IEEEbiographynophoto}

\end{document}